\documentclass[10pt]{article}
\usepackage{algorithm}
\usepackage{lipsum}
\usepackage{algpseudocode}
\usepackage{amsmath,amssymb,latexsym,cite,amsthm}
\usepackage{epsf, pstricks,pgf,xcolor}
\usepackage{epsfig}
\usepackage{subdepth}
\usepackage{bbm}
\usepackage{dsfont}
\usepackage{stmaryrd}
\usepackage{scrextend}
\usepackage{enumerate}

\usepackage[left=1 in,top=1 in,right=1 in,bottom=1 in]{geometry}

\newtheorem{thm}{Theorem}[section]
\newtheorem{theorem}[thm]{Theorem}

\newtheorem{lemma}[thm]{Lemma}

\newtheorem{example}{Example}

\newtheorem{definition}[thm]{Definition}

\newtheorem{remark}{Remark}

\usepackage{caption}
\usepackage{pifont}

\usepackage{array}
\newcolumntype{M}[1]{>{\centering\arraybackslash}m{#1}}
\newcolumntype{N}{@{}m{0pt}@{}}

\date{}
\begin{document}
\title{Storage-Efficient Shared Memory Emulation}

\author{
Marwen Zorgui$^*$, Robert Mateescu$^{**}$, Filip Blagojevic$^{**}$, \\Cyril Guyot$^{**}$, and Zhiying Wang$^*$\\
$^*$CPCC Center, 
University of California, Irvine, $^{**}$Western Digital Research\\
	{  \{mzorgui, zhiying\}@uci.edu}\\
 	{\{robert.mateescu, filip.blagojevic, cyril.guyot\}}@wdc.com
	}
\date{}
\maketitle
\thispagestyle{empty}
%
\begin{abstract}
Improvements in communication fabrics have enabled access to ever larger pools of data with decreasing access latencies, bringing large-scale memory fabrics closer to feasibility. However, with an increase in scale come new challenges. Since more systems are aggregated, maintaining a certain level of reliability requires increasing the storage redundancy, typically via data replication. The corresponding decrease in storage efficiency has led system designers to investigate the usage of more storage-efficient erasure codes. In parallel, storage redundancy introduces consistency challenges that require careful management.

We study the design of storage-efficient algorithms for emulating atomic shared memory over an asynchronous, distributed message-passing system. Our first algorithm is an atomic single-writer multi-reader algorithm based on a novel erasure-coding technique, termed \emph{multi-version code}.
Next, we propose an extension of our single-writer algorithm to a multi-writer multi-reader environment. Our second algorithm combines replication and multi-version code, and is suitable in situations where we expect a large number of concurrent writes. Moreover, when the number of concurrent writes is bounded, we propose a simplified variant of the second algorithm that has a simple structure similar to the single-writer algorithm.

Let $N$ be the number of servers, and the shared memory variable be of size 1 unit. 
Our algorithms have the following properties: (i) 
The write operation terminates if the number of server failures is bounded by a parameter $f$.
The algorithms also guarantee the termination of the read as long as the number of writes concurrent with the read is smaller than a design parameter $\nu$, and the number of server failures is bounded by $f$. 
 (ii) The overall storage size for the first algorithm, and the steady-state storage size for the second algorithm, are all $N/\lceil \frac{N-2f}{\nu} \rceil$ units. Moreover, our simplified variant of the second algorithm achieves the worst-case storage cost of $N/\lceil \frac{N-2f}{\nu} \rceil$, asymptotically matching a lower bound by Cadambe et al. for $N \gg f, \nu \le f+1$.
 (iii) The write and read operations only consist of a small number (2 to 3) of communication rounds. 
 (iv) For all algorithms, the server maintains a simple data structure. A server only needs to store the information associated with the latest value it observes, similar to replication-based algorithms.  
\end{abstract}
\newpage
\section{Introduction}
The emulation of a consistent, fault-tolerant, read-write shared memory in a distributed, asynchronous message-passing network has been an active area of research in distributed computing theory. Several applications demand concurrent and consistent access to the stored value by multiple writers and readers. In their celebrated paper \cite{ABD}, Attiya, Bar-Noy, and Dolev proposed a fault-tolerant algorithm (ABD algorithm) for emulating a shared memory that achieves atomic consistency (linearizability) \cite{herlihy1990linearizability,lamport1986interprocess}. 
ABD uses a replication-based storage scheme at the servers to attain fault
tolerance. In \cite{fan2003efficient}, a two-layer replication based system is also presented, in which one layer is dedicated exclusively to metadata, and the other layer for storage. 
Variations of these algorithms appear in practical systems \cite{lakshman2010cassandra,lynch2002rambo}.

Following \cite{ABD,fan2003efficient}, several papers  
 developed algorithms that use erasure coding instead of replication for fault tolerance, with the goal of improving upon the storage efficiency. 
In erasure coding, each server stores a function of the value called a coded symbol. A decoder can recover the value by accessing a number (called the \emph{coding parameter}) of coded symbols. The number of bits used to represent a coded symbol is typically much smaller than the number of bits used to represent the value. Erasure coding is well known to lead to smaller storage costs as compared to replication \cite{patterson1988case}.
Erasure-code based implementations of consistent data storage appear in \cite{cadambe2017coded, dutta2008optimistic, konwar2016storage, konwar2017layered, spiegelman2016space} for crash failures. 
In \cite{CT, dobre_powerstore, HGR} erasure codes are used in algorithms for implementing atomic memory that tolerate Byzantine failures. 
In \cite{dutta2008optimistic,guerraoui2008collective,konwar2016radon}, the authors provide algorithms that permit repair of crashed servers, while implementing consistent storage. 
Bounds on the performance costs for erasure-code based implementations appear in \cite{cadambe2016information,chockler2017space,spiegelman2016space}.

\noindent\textbf{Contributions}: We consider a distributed message-passing network with fixed $N$ nodes and reliable channels. Nodes can have crash failures.  Up to $f$ server nodes can fail, $f\le (N-1)/2$, and an arbitrary number of client nodes can fail. We first propose a single-writer multi-reader atomic shared memory emulation algorithm, based on which a multi-writer multi-reader algorithm is then presented.

The algorithms guarantee the termination of the read as long as the number of writes concurrent with the read is smaller than a \emph{liveness parameter} $\nu$, and the number of server failures is bounded by $f$. It also ensures the termination of write operations if the number of server failures is no more than $f$.

In the steady state, in all of our algorithms, servers store only a fraction  of the object value. In particular, assume that the value is of size 1 unit.
The total storage size for the single-writer algorithm and the steady-state storage size for the multi-writer algorithm are both $\frac{N}{k}$ units; the worst-case total storage size for the multi-writer algorithm is $k+2f+\frac{N-k-2f}{k}$ units, where $k=\lceil \frac{N-2f}{\nu} \rceil$ is the coding parameter.
When the number of concurrent writes in the system is bounded, we propose a simplified variant of our multi-writer algorithm that has an asymptotic optimal worst-case storage cost matching a lower bound in \cite{cadambe2016information}.

Our algorithms have a simple structure reminiscent of the ABD algorithm as servers only need to store information associated with the latest value they observe, without any logs or history. We call it \emph{in-place update}. The write and read operations only consist of 2 or 3 communication rounds.

The coding parameter in our algorithms is motivated by the work in \cite{wang2017multi}. Traditional erasure codes assume that only one version of the value needs to be encoded and decoded.
In \cite{wang2017multi}, the authors introduced a new family of erasure codes called \emph{multi-version coding}, that allows multiple versions of the value to be present in the system.
Our coding parameter is due to one construction of multi-version codes, enabling the servers to store only a single version and use in-place update.
This also allows us to derive exact statements on the liveness guarantee of read operations. 

From a practical perspective, in  distributed memory- or DRAM-based storage systems such as Memcached \cite{memcached} and RAMCloud \cite{ousterhout2010case}, storage space is costly and management of several versions of the same data object is challenging. Hence, the low storage cost and the simplicity of our algorithms make them attractive in such systems.

\noindent\textbf{Related work}:
Assume that a read operation is concurrent with several writes, including failed writes. Then, in erasure coding, it is possible that the reader obtains information of different values, but does not have sufficient number of coded symbols to decode and return any value. In order to handle the difficulty brought by
concurrent writes, several techniques and liveness guarantees have been proposed, described below.

Algorithms in \cite{cadambe2017coded, dobre_powerstore, HGR} store history of received coded symbols, and hide ongoing writes from a read until enough number of coded symbols have been propagated to the servers. However, the worst-case storage cost grows unbounded with the number of concurrent writes. 
Algorithms in \cite{dutta2008optimistic,konwar2016storage,konwar2017layered} propagate full replicas at a first phase before performing erasure coding at a second phase. In SCCK \cite{spiegelman2016space} a writer communicates full replicas, and each server, upon receipt of a full replica, either stores  a coded symbol or the full replica depending on its state, leading to a worst-case storage of $2N$ units.
ORCAS-A \cite{dutta2008optimistic} is similar to our algorithms in that the server stores only the latest version. 
However, the read operation uses reader registration to be explained later. 
The algorithm in \cite{konwar2016storage} achieves the lowest overall storage in the steady state at the expense of costly write communication on the order of $N^2$ units. 
In \cite{konwar2017layered}, replicas of all ongoing writes are stored in an edge layer of servers, and coded symbols are stored in a back-end layer. The total worst-case storage can be unbounded even with garbage collection in the edge layer.

The strongest liveness guarantee for read operations is \emph{wait-freedom}, which guarantees that a non-failed process completes its execution irrespective of the actions of other processes, implemented by \emph{reader registration} in \cite{CT, dutta2008optimistic, konwar2016storage, konwar2017layered}. That is, a read operation registers itself at the servers it contacts, and keeps receiving symbols from them until successful recovery of a value. However, the amount of communication of a read operation can be unbounded and depends on concurrent writes. In contrast, our algorithms guarantee liveness if the number of concurrent writes with a read is smaller than $\nu$, and only uses 2 or 3 rounds of communications. A similar liveness setting is found in CASGC \cite{cadambe2017coded}, but we will demonstrate the advantage of our algorithms in Section \ref{sec:discussion} in terms of the storage cost and protocol simplicity. 
SCCK \cite{spiegelman2016space} satisfies the \emph{finite-write termination} liveness, namely, in every execution with finitely many writes, every read operation invoked by a non-failed reader terminates \cite{abraham2006byzantine}.
In HGR \cite{HGR}, read operations satisfy \emph{obstruction-freedom}, that is, a read returns if there is a sufficiently long period during the read when no other operation takes steps. 

{\bf Organization:} In Section \ref{sec:preliminaries}, we give an overview of the algorithms and introduce useful definitions and lemmas. The single-writer and multi-writer algorithms are presented and analyzed in Sections \ref{Section_Single_Writer} and \ref{multi_writer_algorithm_1}. Detailed comparisons with previous algorithms and conclusions are shown in Section \ref{sec:discussion}. 

\section{Preliminaries and Overview of Algorithms}
\label{sec:preliminaries}
We study the emulation of a shared atomic memory in an asynchronous message-passing network. We assume a single data object without loss of generality.  The number of server nodes is denoted as $N$. The number of client nodes can be unbounded. All the client and server nodes are connected by point-to-point reliable channels, and a node failure is assumed to be a crash failure. 
Every new invocation at a client waits for a response of a preceding invocation at the same client (called well-formedness). 
We require the following safety and liveness properties, irrespective of the number of client failures.

$\bullet$ \emph{Atomicity:} The algorithm must emulate a shared atomic read-write object that supports concurrent access by the clients 
in the system, where the observed global external behaviors {``look like''} the object is being accessed sequentially \cite{lamport1986interprocess}.

$\bullet$ \emph{$\nu$-concurrency wait-freedom:}  
We require a write operation to terminate if the number of server failures in the execution is bounded by a parameter $f$, and a read operation to terminate if the number of server failures is bounded by $f$ and the number of concurrent writes is less than a parameter $\nu$. We call such liveness property $\nu$-concurrency wait-freedom, and $\nu$ the liveness parameter.

In practice, our algorithm does not need to know the exact worst-case concurrency level over all executions. Instead, it can use $\nu$ as an estimate of the concurrency, say, for 90\% of the read operations. If a reader is not able to return the value, it can re-try and complete the read if the number of current writes reduces to less than $\nu$.

We define a quorum set $Q$ to be a subset of the server nodes, such that its size satisfies $ |Q| \geq N-f$. It follows that for any two quorums $Q_1,Q_2$, we have $ | Q_1 \cap Q_2 | \geq N-2f$.
We assume that every data value comes from a finite set $\mathcal{V}$. In this paper we refer to $\log_2|\mathcal{V}|$ as 1 unit.
We also arbitrarily choose $v_0$ from $\mathcal{V}$ to be a default value. 
Different versions of the data value are associated with different \emph{tags}. 
We say that a  tag $t$ is \emph{decodable} if the read operation can recover the value corresponding to tag $t$. 
 We let $\Phi$ be an $(N,k)$ maximum distance separable code (e.g. Reed-Solomon code) that takes a value in $\mathcal{V}$ as input and outputs $N$ coded symbols in $\mathcal{W}$, where $\log_2|\mathcal{W}| = \frac{1}{k} \log_2|\mathcal{V}|$, corresponding to $\frac{1}{k}$ unit. Any $k$ of the $N$ coded symbols suffice to decode the value. 
 We set the coding parameter $k$ to be $k=\lceil\frac{N- 2 f}{\nu}  \rceil$. The choice of $k$ is motivated by multi-version codes \cite{wang2017multi}.

\begin{remark}
\label{choice_nu}
For fixed design parameters $f,\nu,N$, the coding parameter is $k=\lceil\frac{N- 2 f}{\nu}  \rceil$. In fact, we can use only $\widetilde{N}=(k-1)\nu+2f+1 \le N$ nodes and do not use the remaining nodes, while keeping the same the coding parameter $\lceil\frac{\widetilde{N}- 2 f}{\nu}  \rceil=\lceil\frac{N- 2 f}{\nu}  \rceil$. Throughout the paper, we will use the reduced number of nodes, and assume the integer $k$ satisfies
\begin{equation}
k= \lceil\frac{N- 2 f}{\nu}  \rceil=1 + \frac{N - (2f+1)}{\nu}.
\label{N_value}
\end{equation}
\end{remark}
\noindent\textbf{Storage and communication cost definitions}: The \emph{storage cost} of an algorithm is defined to be the overall storage size of all the servers. 
In the algorithms that we formulate, each server node stores a list of pairs each of the form $(t, w)$, where  $t$ is a tag, and $w \in \mathcal{V} \cup \mathcal{W}$ depends on the value with tag $t$. In our analysis of the storage cost, we neglect the cost of the tags and other metadata; so the storage cost of an algorithm is measured as the size of $w$'s. 
We define a \emph{steady-state} point in an execution to be a point for which there is no ongoing write, and the completed writes have delivered their messages to all live servers. The steady-state storage cost is the storage cost of a steady-state point. The \emph{worst-case} storage cost corresponds to the largest storage cost among all points in all executions.
The \emph{communication cost} of a read (or write) is defined to be the largest total number of communicated bits associated with the data value over the network, among all read (or write) operations of all executions. Metadata bits are again neglected.

\noindent\textbf{Algorithms overview}:\\
$\bullet$ Algorithm \ref{sw_algo}. In Section \ref{Section_Single_Writer}, we describe a single-writer multi-reader algorithm, referred to as Algorithm \ref{sw_algo}. 
The write operation has one phase, where the value is encoded using an $(N,k)$ erasure code and propagated to at least $N-f$ servers.  
The read operation is carried out in two phases, one for getting values, and one for writing-back the decoded value. 

\noindent $\bullet$ Algorithm \ref{Algorithm_mw_1}. In Section \ref{multi_writer_algorithm_1}, we extend Algorithm \ref{sw_algo} to the multi-writer multi-reader setting, referred to as Algorithm \ref{Algorithm_mw_1}. The write protocol of Algorithm \ref{Algorithm_mw_1} has a \emph{pre-write} phase in which full replicas are propagated to at least $k+f$ servers, followed by a \emph{finalize} phase, where coded symbols replace the replicas in a quorum of servers of size $N-f$. 
Moreover, in \cite{cadambe2016information} a lower bound on the storage cost is developed under the same liveness condition as our algorithms.  We show that a variant of Algorithm \ref{Algorithm_mw_1}, referred to as Algorithm \ref{Algorithm_mw_1}-A, is essentially storage-optimal, based on the lower bound of \cite{cadambe2016information}.

The storage and communication costs of our algorithms, ABD, and two previous coding-based algorithms are shown in Table \ref{tab:contribution}. The two coding-based algorithms are listed because they employ somewhat similar protocol structures as ours. The parameter $\nu$ in our algorithms illustrates the tradeoff between liveness of read operations and the storage size. The smaller $\nu$ is, the smaller the storage size is, but the smaller the number of concurrent writes that a successful read can tolerate. In particular, when $\nu \geq N-2f$, our algorithms reduce to ABD.
We compare the storage of multi-writer algorithms. Assume $\nu<2f+1, N \gg f$. Then, in descending order of the worst-case storage cost, we have CASGC, SCCK, ABD, and Algorithm \ref{Algorithm_mw_1}.
 In descending order of the steady-state storage cost, we have ABD, CASGC, Algorithm \ref{Algorithm_mw_1}, and SCCK.
More detailed discussions can be found in Section \ref{sec:discussion}.

 \begin{center}
\begin{tabular}{|c|c|c|c|c|}
\hline
~  & worst-case  & steady-state  & write  & read  \\
~ & storage & storage & communication & communication \\
\hline
Alg. \ref{sw_algo},  \ref{Algorithm_mw_1}-A  &  $\frac{N}{k}$ & $\frac{N}{k}$    &  $\frac{N}{k}$ & $\frac{2N}{k}$  \\
\hline
Alg. \ref{Algorithm_mw_1} &  $k+2f + \frac{N-k-2f}{k} $ &  $\frac{N}{k}$ &  $ k+2f + \frac{N-k-2f}{k}$ &  $2(k+2f + \frac{N-k-f}{k})$   \\
\hline
ABD \cite{ABD} & $2f+1$ & $2f+1$ & $2f+1$ & $2(2f+1)$ \\
\hline
CASGC \cite{cadambe2017coded} & unbounded & $\frac{\nu N}{N-2f}$ & $\frac{ N}{N-2f}$ & $\frac{2 N}{N-2f}$ \\
\hline
SCCK \cite{spiegelman2016space} & $2N$ & $\frac{ N}{N-2f}$ & $N$ & $2N$\\
\hline
\end{tabular}
 \captionof{table}{Summary of the storage and write communication costs of the different algorithms. Here $k=\lceil\frac{N- 2 f}{\nu}  \rceil$. The size of a data value is 1 unit. 
 For the worst-case storage, we assume an arbitrary number of concurrent writes in an execution.
 Assume our algorithms use liveness parameter $\nu$, ABD uses only $2f+1$ servers, GASGC uses parameters $(k_{CASGC},$ $\delta)= (N-2f,\nu-1)$, and SCCK uses coding parameter $k_{SCCK}=N-2f$.
 ABD satisfies the strongest liveness which is wait-freedom; our proposed algorithms and CASGC satisfy $\nu$-concurrency wait-freedom; SCCK has weaker liveness, namely, finite-write termination. 
\label{tab:contribution} }
\end{center}

Next we state a lemma of a sufficient condition for atomicity, which will be used to prove correctness for our algorithms.

\begin{lemma}[Lemma 13.16 of \cite{Lynch1996}]
	\label{lem:partialorder}
	{Let $\beta$ denote of a sequence of actions of the external interface of a read/write object. Suppose $\beta$ is well formed for each client and contains no incomplete operations. Let $\Pi$ be the set of all operations in $\beta$.}
A sufficient condition for atomicity of $\beta$ is: there exists a partial ordering $\prec$ of all the operations in $\Pi$, satisfying the following properties:\\
(1) If the response for $\pi_1$ precedes the invocation for $\pi_2$ in $\Pi$, then it cannot be the case that $\pi_2 \prec \pi_1$. \\
(2)  If $\pi_1$ is a write operation in $\Pi$ and $\pi_2$ is any operation in $\Pi$, then either $\pi_1 \prec \pi_2$ or $\pi_2 \prec \pi_1$. \\
(3) The value returned by each read operation is the value written by the last preceding write operation according to $\prec$ (or the default value, if there is no such write).
\end{lemma}
We now define the partial ordering that we use in conjunction with Lemma \ref{lem:partialorder} in the correctness proofs. We define tags of operations for each algorithm in its corresponding section.
\begin{definition}[Partial Ordering $\prec$]
	\label{def:partialorder}
Consider an execution $\alpha$ and consider two operations $\pi_1, \pi_2$ that complete in $\alpha$. Let $T(\pi_1)$ and $T(\pi_2)$ respectively denote the tags of operations $\pi_1$ and $\pi_2$. Then we define the partial ordering on the operations as: $\pi_1 \prec \pi_2$ if \\
(1)  $T(\pi_1) < T(\pi_2)$; or \\
(2)   $T(\pi_1) =  T(\pi_2)$ for write $\pi_1$ and read $\pi_2$.
\end{definition}
\begin{algorithm}
\caption{\textbf{: single-writer setting} }
\begin{algorithmic}  [1]
\Statex  \textbf{Write protocol} \newline
state variable: Tag $t$,  $t \in \mathbb{N}$\newline
 {\emph{initial state:} Tag $0$.} \newline
 \emph{Input:} Value $v$, {$v \in \mathcal{V}$}.
\State	Increment the state, that is set $t\leftarrow t+1$. \label{increment_tag_sw}
\State 	Use the $(N,k)$ code to get $N$ coded symbols. Denote $(y_1, y_2,\ldots, y_N) = \Phi(v)$. 
 \State Send $put(t,y_s)$ to server $s,$ for every $s \in \{1,2,\ldots,N\}$. Await acknowledgement from a quorum, and then terminate.
\newline  
\Statex \textbf{Read protocol}
\State   Send query request $get$ to all servers, await pairs $(t,code)$ from a quorum.  
\State   Let $R$ be the set of response pairs. 
\State Let $T$ be the set of decodable tags $t$ occurring in $R$ such that \label{read_start_alg_1}  
\State  $(i)$ $t$ has at least $f+1$ coded symbols,  \label{algo1_line_1}
\State $(ii)$ Or, the number of tags strictly higher than $t$ is at most $\nu$.
 \label{algo1_line_2}
     \If{ $T \neq \emptyset$} 
     \State  Let $t = \max(T)$, and $v$ its value.  \label{read_finish_alg_1}
     		 \State  \emph{write}$\_$\emph{back}$(t,v)$    		 
     \Else
     \State $\_$\emph{abort}$\_$
     \EndIf    
 \Procedure{ \emph{write}$\_$\emph{back} }{$ t,v $}
\State		  Let $(y_1,y_2,\ldots, y_N) = \Phi(v)$.
 \State		  {Propagate} $put(t,y_s)$ to server $s$ for every $s \in \{1,2,\ldots,N\}$, await acknowledgement from a quorum, and then terminate by returning $v$ .
\EndProcedure
\newline
\Statex \textbf{Server $s$ protocol}
      \Statex \emph{state variable:} A pair $(t,{y})$, where {$t \in \mathbb{N}$, $y \in \mathcal{W}$.}
	  \Statex {\emph{initial state:} Store the default pair $(0, y_s)$, where $y_s$ is the $s$th component of $\Phi(v_0),$ where $v_0 \in \mathcal{V}$ is the default initial value.}
    \State On receipt of \emph{get}: respond with $(t,{y})$.
   \State On receipt of {$put(t_{new},y_{new})$}: If $t_{new}>t$, then set {$t \gets t_{new}$ and ${y} \gets y_{new}.$} In any case respond with acknowledgement.
\end{algorithmic}
 \label{sw_algo}
\end{algorithm}
\section{Single-Writer Algorithm}
\label{Section_Single_Writer}
\subsection{Algorithm Description}
In this section, we describe our single-writer multi-reader algorithm (See Algorithm \ref{sw_algo}).
The different phases of the write and read protocols are executed sequentially. 
In each phase, a client sends messages to servers to which the non-failed servers respond. Termination of each phase depends on getting responses from at least one quorum.
In this algorithm, the write protocol increments the tag, and writes the value to a quorum of servers using the erasure code. 
The read protocol has an $\_ abort \_$ internal action. In case the $\_ abort \_$ action is invoked, the client does not return and the operation which invokes it does not terminate. The action indicates to the reader that the concurrency bound was violated causing the read not to terminate. From the viewpoint of a practical storage system, we note that the $\_abort \_$ action can prompt the reader to invoke a read request again; however, we do not formally incorporate such an invocation in the description of Algorithm \ref{sw_algo}. We comment on the possibility of invoking multiple rounds of read briefly in Section \ref{sec:liveness}. The read protocol has a procedure called \emph{write}$\_$\emph{back}, which is triggered whenever the reader can recover a certain value $v$ from the responses and safely return it. 
The read returns a value with tag $t$ that is decodable and also satisfies some conditions, as specified by Lines \ref{read_start_alg_1} through \ref{read_finish_alg_1} in Algorithm \ref{sw_algo}.
For a tag-value pair $r=(t, v)$, we write $tag(r)=t$.
\begin{remark}
\label{ABD_equivalence}
If $\nu \geq N-2f$, the code used in Algorithm \ref{sw_algo} specializes to the replication-based ABD algorithm. We have $k=1$, i.e., every server stores a full replica. Moreover, the reader recovers the value corresponding to the highest tag observed among the responses and the value satisfies the condition in Line \ref{algo1_line_2}. 
\end{remark}
Throughout the section, we assume that $\nu \geq 2$ and $k>1$.
\subsection{Safety Properties}
\label{sec:safety}
In this section, we present safety properties satisfied by Algorithm \ref{sw_algo}. 
We first show in Lemma \ref{data_persistence_sw} that there always exists some value that can be recovered from the servers during the execution of Algorithm \ref{sw_algo}. 
Then, we show that Algorithm \ref{sw_algo} emulates an atomic shared memory in Theorem \ref{thm:atomicity},
using Lemma \ref{lem:partialorder} on the partial ordering $\prec$ in Definition \ref{def:partialorder}. To show this, we  prove Lemmas \ref{lem:a_safety_property}, \ref{lem:tagsordering_read}, \ref{lem:tagsordering_write} and \ref{lem:writeTagDifferentAlg1}.
We finally prove a safety property in Lemma \ref{lem:a_secondsafety_property} that will be used in Section \ref{sec:liveness} where we describe liveness properties.

{We now define the \emph{tag of an operation}. In the definition, we use the fact that every read or write operation that completes propagates \emph{put} messages to the servers with a particular tag.} Note that the tag of an operation is defined for every operation that completes in an execution. Furthermore, the tag is not defined for read operations that abort, since these operations do not propagate a \emph{put} message and are not considered complete.
\begin{definition}[Tag of an operation $\pi$]
\label{definition_tags_single_writer}
	{ Let $\pi$ be an operation in an execution $\alpha$.} {The tag of operation $\pi$ is defined to be the tag associated with the \emph{put} messages that the operation propagates to the servers.} 
\label{def:tag_definition}
\end{definition}
\begin{lemma}[Persistence of data]
\label{data_persistence_sw}
The value written by either the latest complete write or a newer write is available from every set of at least $N-f$ servers.
\end{lemma} 
\begin{proof}
Let $Q_w$ denote the quorum of servers that replied to the \emph{put} message of the latest finished write $\pi_w$ (if no write has finished in the execution, $Q_w$ can be quorum from the set of live of servers, and $T(\pi_w)$ is assumed to be 0). Thus, each server in $Q_w$ has a tag that is at least as large as $T(\pi_w)$. Because there is at most a single ongoing incomplete write operation for a single writer, the number of tags in $Q_w$ is at most 2. Thus, one of the tags, say $t$, appears in at least $\lceil \frac{|Q_w|}{2} \rceil\geq \lceil \frac{N-2f}{\nu} \rceil=k$ servers and $t \geq T(\pi_w)$. Therefore, the value corresponding to $t$ is available in the system.
\end{proof}
\begin{remark}
\label{instantaneous_image}
Lemma \ref{data_persistence_sw}  implies that a read operation $\pi_r$ can decode and return a value satisfying Line \ref{algo1_line_2}, if the reader gets responses that corresponds to the stored $(tag, element)$ pairs of a quorum at some point $P$ (we call it an instantaneous image).
\end{remark}
\begin{lemma}\label{lem1}
Consider any execution $\alpha$ of the algorithm and consider a write or read operation $\pi_1$ that completes in $\alpha$. Let $T(\pi_1)$ denote the tag of the operation $\pi_1$ and let $Q_1$ denote the quorum of servers from which responses are received by $\pi_1$ to its \emph{put} message.  Consider a read operation $\pi_r$ in $\alpha$ that is invoked after the termination of the write operation $\pi_1$. Suppose that the read $\pi_r$ receives responses to its \emph{get} message from a quorum $Q_r$. Then, \\
(1) Every server $s$ {in} $Q_1 \cap Q_r$ responds to the \emph{get} message from $\pi_r$ with a tag that is at least as large as $T(\pi_1)$.\\
(2) If, among the responses to the \emph{get} message of $\pi_r$ from the servers in $Q_{1} \cap Q_{r}$, the number of tags is at most {$\nu$}, then there is some tag $t$ such that 
 (i) $t \geq T(\pi_1)$, and 
  (ii) from the servers in {$Q_1 \cap Q_r$}, operation $\pi_r$ receives {at least} $k$ responses to its \emph{get} message with tag $t.$ 
  \label{lem:a_safety_property}
\end{lemma}
\begin{proof}
\emph{Proof of (1).}
Consider any server $s$ in $Q_1 \cap Q_r$. From the server protocol we note that at every point after the reception of $\pi_1$'s \emph{put} message, it stores a tag that is no smaller than $T(\pi_1)$. 
So it responds to the \emph{get} message with a tag that is at least as large as $T(\pi_1)$. This completes the proof of (1).

\noindent \emph{Proof of (2).}
Among the responses from $Q_1 \cap Q_r$, the read $\pi_r$ receives at most $\nu$ different tags. By {the} Pigeonhole principle, there is at least one tag $t$ such that it receives at least  $\lceil \frac{|Q_1 \cap Q_r|}{\nu} \rceil$ responses with $t$. Since $|Q_{1} \cap Q_{r}| \geq N-2f$, we infer that the operation $\pi_r$ receives at least $\lceil\frac{N-2f}{\nu}\rceil = k$ responses with tag $t$. From (1), we infer that $t \geq T(\pi_1)$ to complete the proof.
\end{proof}
\begin{lemma}
	{Consider an execution $\alpha$ of Algorithm \ref{sw_algo}. Let $\pi_1$ be a write or read operation that completes in $\alpha$, and let $\pi_2$ be a read operation that completes in $\alpha$. }Let $T(\pi_1)$ denote the tag of operation $\pi_1$ and $T(\pi_2)$ denote the tag of operation $\pi_2$. If $\pi_2$ begins after the termination of $\pi_1$, then $T(\pi_2) \geq T(\pi_1)$.
	\label{lem:tagsordering_read}
\end{lemma}
\begin{proof}
Let $Q_1$ denote the quorum that responds to $\pi_1$'s \emph{put} message. Let $Q_2$ denote the quorum that responds to $\pi_2$'s \emph{get} message. Let $Q=Q_1 \cap Q_2.$ 
We prove the claim by contradiction. Suppose that $T(\pi_2) < T(\pi_1)$. Either Line ~\ref{algo1_line_1} or Line \ref{algo1_line_2} should be satisfied so that $\pi_2$ completes.

\noindent From (1) in Lemma \ref{lem:a_safety_property}, we infer that every server in $Q$ responds to the \emph{get} message of $\pi_2$ with a tag that is at least as large as $T(\pi_1)$. Because $T(\pi_1) > T(\pi_2)$, the value returned by $\pi_2$ must have been obtained using the responses from servers in $Q_2 \backslash Q_1 $. Thus $| \{u=T(\pi_2) | u= tag(r), \textrm{for some } r \in R\}| \le |Q_2 \backslash Q_1| \le N-(N-f) = f $. Thus Line \ref{algo1_line_1} cannot be satisfied.

\noindent Assume Line \ref{algo1_line_2} is satisfied. Because every server in $Q$ responds with a tag that it is at least as large as $T(\pi_1)$, which is greater than $T(\pi_2)$, and because $Q \subseteq Q_2$, we infer that the number of distinct response tags from $Q$ that are larger than $T(\pi_2)$ is at most $\nu$. 
Property (2) in  Lemma \ref{lem:a_safety_property} implies that there exists a tag $t \geq T(\pi_1)$ that appears in at least $k$ responses from $Q$. From the read protocol, we infer that the tag of the read operation should be at least $t$. That is, $T(\pi_2) \geq t$. But we know that $t \geq T(\pi_1) > T(\pi_2)$, which is a contradiction.
\end{proof}
 \begin{remark}
 \label{multiple_return_values}
There can be multiple values that can be safely returned by a read operation. Indeed, as can be inferred from the proof of Lemma \ref{lem:tagsordering_read}, any value satisfying Line \ref{algo1_line_1} in Algorithm \ref{sw_algo} can be returned safely, even if a higher tag can be recovered by the read operation.
 \end{remark}  
 \begin{remark}
\label{system_simplification}
If $k>f$, the reader protocol is simplified so that Lines \ref{algo1_line_1} and \ref{algo1_line_2} are omitted. Indeed, any decodable tag has at least $k \geq f+1$ coded symbols, which automatically satisfies Line \ref{algo1_line_1} in Algorithm \ref{sw_algo}, and can be safely returned because of Remark \ref{multiple_return_values}.
\end{remark}

\begin{lemma}
Consider an execution $\alpha$ of Algorithm \ref{sw_algo}. Let $\pi_1$ be a write or read operation that completes in $\alpha$, and $\pi_2$ be a write operation that completes in $\alpha$.  Let $T(\pi_1)$ denote the tag of operation $\pi_1$ and $T(\pi_2)$ denote the tag of operation $\pi_2$.  If $\pi_2$ begins after the termination of $\pi_1$, then $T(\pi_2) > T(\pi_1)$. 
	\label{lem:tagsordering_write}
\end{lemma}
\begin{proof}
We first consider the case where $\pi_1$ is a write. Later, we consider the case where $\pi_1$ is a read.

	\emph{Case 1:} If $\pi_1$ is a write operation, then from the write protocol, we note that the state of the writer at any point after the completion of $\pi_1$ is at least as large as $T(\pi_1)$. Since $\pi_2$ begins after the termination of $\pi_1,$ and since $\pi_2$ increments the client state to obtain $T(\pi_2)$, we infer that $T(\pi_2)$ is strictly larger than the client state at the point of invocation of $\pi_2$, which is at least as large as $T(\pi_1)$. Therefore,  $T(\pi_2) > T(\pi_1)$.

	\emph{Case 2:} If $\pi_1$ is a read operation, note that the tag $T(\pi_1)$ of the operation corresponds to the tag sent by some server $s$ as a part of its message. From the server protocol, we note that the tag $T(\pi_1)$ was obtained by the server $s$ by a message from some write operation $\pi$, or from the default value. If from the default value, then from Line \ref{increment_tag_sw}, the result follows. Otherwise, we note that $\pi$ {begins} at the {writer} before the termination of operation $\pi_1$. Because $\pi_2$ begins after the termination of $\pi_1,$ {and because there is a single writer,} we note that $\pi_2$ begins after the termination of $\pi$ at the writer. From the argument presented in Case 1, we infer that the tag of operation $\pi_2$ is strictly larger than the tag of operation $\pi$. The tag of operation $\pi$ is equal to $T(\pi_1)$. Therefore, we have $T(\pi_2) > T(\pi_1)$.   This completes the proof.
\end{proof}

\begin{lemma}
Let $\pi_1, \pi_2$ be write operations that terminate in an execution $\alpha$ of Algorithm \ref{sw_algo}. Then, $T(\pi_1)\neq T(\pi_2)$.
\label{lem:writeTagDifferentAlg1}
\end{lemma}
\begin{proof}
Since $\pi_1$ and $\pi_2$ are invoked at the same client, there are only two possibilities: either $\pi_2$ begins after $\pi_1$ terminates, or $\pi_1$ begins after $\pi_2$ terminates. From Line \ref{increment_tag_sw} in Algorithm \ref{sw_algo}, Based on Lemma \ref{lem:tagsordering_write}, we infer that it is the case that either $T(\pi_1) > T(\pi_2)$, or $T(\pi_2) > T(\pi_1)$.
\end{proof}
The following theorem states the main result on atomicity, and the proof follows from Definitions \ref{def:partialorder} and \ref{definition_tags_single_writer}, combined with Lemmas \ref{lem:partialorder}, \ref{lem:tagsordering_read}, \ref{lem:tagsordering_write}, and \ref{lem:writeTagDifferentAlg1}. 
\begin{theorem}
Algorithm \ref{sw_algo} emulates an atomic read-write object.
  \label{thm:atomicity}
\end{theorem}
\begin{proof}
Let $\beta$ denote a sequence of actions of the external interface of a read/write object satisfying the conditions in Lemma \ref{lem:partialorder}.  Let $\Pi$ be the set of operations in $\beta$. Note that because $\Pi$ consists of operations that complete, every operation in $\pi$ has a tag. Definition \ref{def:partialorder} imposes a partial order $\prec$ on the set $\Pi.$ Let $\pi_1$ and $\pi_2$ be two operations in $\Pi$. Let $T(\pi_1)$ denote the tag of operation $\pi_1$ and $T(\pi_2)$ denote the tag of operation $\pi_2$.  We show that operations $\pi_1$ and $\pi_2$ satisfy Properties {\bf(1)}, {\bf(2)} and {\bf{(3)}} of Lemma \ref{lem:partialorder}.

	\emph{Proof of {\bf{(1)}:}} We consider two cases. First, we consider the case where $\pi_2$ is a read. Second, we consider the case where $\pi_2$ is a write. In the first case, if $\pi_2$ is a read, then Lemma \ref{lem:tagsordering_read} implies $T(\pi_2) \geq T(\pi_1)$. As per Definition \ref{def:partialorder}, we infer that it is not the case that $\pi_2 \prec \pi_1$. In the second case, if $\pi_2$ is a write, then Lemma \ref{lem:tagsordering_write} implies that $T(\pi_2) > T(\pi_1)$. As per Definition \ref{def:partialorder}, we infer that it is not the case that $\pi_2 \prec \pi_1$.

	\emph{Proof of {\bf{(2)}:}} Recall that $\pi_1$ is a write operation. First consider the case where $\pi_2$ is a read operation. There are only two possibilities: either $T(\pi_1) > T(\pi_2)$, then $\pi_2 \prec \pi_1$. Otherwise, $\pi_1 \prec \pi_2$. Now consider the case where $\pi_2$ is write operation. From lemma \ref{lem:writeTagDifferentAlg1}, either $T(\pi_1) > T(\pi_2),$ which implies that $\pi_2 \prec \pi_1$, otherwise, $T(\pi_2) > T(\pi_1)$, which implies that $\pi_1 \prec \pi_2$. This completes the proof of {\bf{(2)}}.

\emph{Proof of {\bf{(3)}:}} Consider a read operation $\pi_1$ that returns value $v$. Let $T(\pi_1)$ denote the tag of the read operation. 

We first consider the case of $T(\pi_1)=0$. By Lemma \ref{lem:tagsordering_read} and since the tag of a write operation is greater or equal to 1 by Line \ref{increment_tag_sw}, there are no writes preceding $\pi_1$. By the read protocol, the read receives at least $k$ responses with codeword symbols obtained by applying the encoding function $\Phi$ on the default value $v_0$, decodes and returns $v_0$.

We now consider the case of  $T(\pi_1)>0$. Then, from the server protocol, we note that a write operation $\pi_2$ encoded a value $w$ with codeword symbols corresponding to tag $T(\pi_1).$ From our definition of the tag of an operation, we note that the tag of the operation $\pi_2$ is equal to $T(\pi_1)$. From our definition of partial order, we note that $\pi_2$ is the last write operation that precedes $\pi_1$ as per $\prec$. To complete the proof, we need to show that $\pi_1$ returns $w$. From the read protocol, we note that the read operation receives at least $k$ messages from $k$ distinct servers with tag $T(\pi_1),$ and the corresponding codeword symbols. From the write and server protocols, we infer that these $k$ codeword symbols were obtained by applying the $(N,k)$ code to $w$. Therefore, the reader decodes $w$ and returns. This completes the proof.
\end{proof}

Next, we use Lemma \ref{lem:a_safety_property} to show a safety property that will be used later to prove liveness. 
\begin{lemma} \label{lem2}
	Consider any execution $\alpha$ of Algorithm \ref{sw_algo}. {Let $\pi_r$ denote a read operation in $\alpha$ that receives a quorum $Q_r$ of responses to its \emph{get} message.} Let $\mathcal{S}$ denote the set of all writes that terminate before the invocation of $\pi_r$ in $\alpha$. {If $\mathcal{S}$ is non-empty,} let $t_w$ denote the largest among the tags of the operations in $\mathcal{S}$. {If $\mathcal{S}$ is empty, let $t_w=0$}. 

If the number of writes concurrent with the read $\pi_r$ is smaller than $\nu$, then there is some tag ${t}$ such that  
\\ (1) ${t} \geq t_w$, 
\\ (2) $\pi_r$ receives at least $k$ responses to its \emph{get} message with tag ${t}$, and
\\ (3) the number of tags that are higher than $t$ is smaller than $\nu$.
\label{lem:a_secondsafety_property}
\end{lemma}
\begin{proof}
We first argue that, among the responses to the read's \emph{get} message from $Q_r$, the reader gets fewer than $\nu$ distinct response tags that are larger than $t_w$. By definition of the tag $t_w$, if the read receives a tag $t$ that is larger than $t_w$, then the tag $t$ corresponds to the tag of a write operation $\pi$ that is concurrent with $\pi_r$. Since the number of writes that are concurrent with the read is smaller than $\nu$, the read receives fewer than $\nu$ distinct tags that are larger than $t_w$. 
	
We assume that $\mathcal{S}\neq \emptyset$. The case $\mathcal{S}= \emptyset$ can be treated in a similar way.
Consider the write operation $\pi_w$ in $\mathcal{S}$ whose tag is $t_w$. Let $Q_w$ denote the quorum of servers from which responses were received by the writer to the \emph{put} message of $\pi_w$. Property (1) in Lemma \ref{lem:a_safety_property} implies that every server $s$ in $Q_w \cap Q_r$ responds with a tag that is at least as large as $t_w$. Note that we have already shown that $\pi_r$ receives fewer than $\nu$ distinct tags to its \emph{get} message that are larger than $t_w$. That is, among the responses received by $\pi_r$ from the servers in $Q_w \cap Q_r$ to its \emph{get} message, there are at most $\nu-1$ distinct tags that are larger than $t_w$. Since some of the servers in $Q_w \cap Q_r$ may respond with tag $t_w$, we infer that among the responses received by $\pi_r$ from servers in $Q_w \cap Q_r$, there are at most $\nu$ distinct tags. Property (2) of Lemma \ref{lem:a_safety_property} implies the statement of the lemma, since it implies that there is at least one tag ${t}$ which is no smaller than $t_w$ such that at least $k$ responses with tag ${t}$ are received by $\pi_r$. It can be seen that (3) holds.
\end{proof}
\subsection{Liveness Properties}
\label{sec:liveness}
We state the liveness of Algorithm 1.  
Recall that we focus here on the single-writer algorithm. 
\begin{theorem}[Termination of writes]
Consider any fair execution $\alpha$ of Algorithm \ref{sw_algo} where the number of server failures is at most $f$, and the write client does not fail. Then, every write operation terminates in $\alpha$.
\label{thm:writecompletes}
\end{theorem}
\begin{proof}
Consider a fair execution $\alpha$ and let $\pi$ denote an arbitrary write operation in $\alpha$. Consider a non-failing server $s$ in $\alpha$. In a fair execution, eventually server $s$ receives the \emph{put} message of operation $\pi$. From the server protocol, we note that server $s$ responds to the \emph{put} message of the write with an acknowledgement. Therefore, in a fair execution, eventually operation $\pi$ receives an acknowledgement from every non-failing server $s$. Since the number of server failures is no bigger than $f$, there is at least one quorum $Q$ consisting entirely of non-failing servers. Therefore, the write operation $\pi$ receives acknowledgments from at least one quorum of servers. From the write protocol, we infer that operation $\pi$ terminates. This completes the proof.
\end{proof}
\begin{theorem}[Termination of reads]
Consider any fair execution $\alpha$ of Algorithm \ref{sw_algo} where the number of server failures is at most $f$. Consider any read operation that is invoked at a non-failing client in $\alpha$. If the number of writes concurrent with the read is strictly smaller than $\nu$, and the read client does not fail, then the read operation completes in $\alpha$.
\label{thm:readcompletes}
\end{theorem}
\begin{proof}
Consider a read operation $\pi_r$ in $\alpha$ such that the number of writes that are concurrent with $\pi_r$ is smaller than $\nu$. We show that $\pi_r$ completes in $\alpha$. Since the number of server failures in $\alpha$ is at most $f$, we note that $\pi_r$ receives responses from a quorum $Q_r$ of servers to its \emph{get} message.	To show completion of $\pi_r$, we show that

 \noindent(1) there is a tag that is decodable, and \\
  \noindent (2) it satisfies either Line \ref{algo1_line_1} or Line \ref{algo1_line_2}, and the \emph{write}$\_$\emph{back} phase terminates, that is, \\
 \noindent  (3) the read receives a quorum of acknowledgments to its \emph{put} message.

{Since $\alpha$ is an execution where there are at most $f$ failures, there is at least one quorum consisting entirely of non-failing servers. Since every server eventually responds to the \emph{get} message of the read, the read gets responses from some quorum of servers $Q_r$.}  {Lemma \ref{lem:a_secondsafety_property} implies that there is a tag $\overline{t}$ such that the read $\pi_r$ receives at least $k$ responses with tag $\overline{t}$. In other words, Lemma \ref{lem:a_secondsafety_property} implies (1).}
  
We now show (2). We distinguish two cases. If Line \ref{algo1_line_1} is satisfied, then (2) follows. Otherwise, by virtue of Lemma \ref{lem:a_secondsafety_property} (iii), Line \ref{algo1_line_2} is satisfied and (2) follows.

  (3) follows from the fact that the number of server failures in $\alpha$ is at most $f$, and therefore there is at least one quorum of servers that eventually responds to the \emph{put} message from the read $\pi_r.$ We have thus shown (1), (2) and (3), which imply that the read operation $\pi_r$ terminates.
\end{proof}
\begin{remark}
\label{sufficient_liveness_condition}
The condition of concurrency being smaller than $\nu$ in Theorem \ref{thm:readcompletes} is a sufficient condition for the termination of a read operation, but it is not necessary. Indeed, as highlighted in Remark
\ref{multiple_return_values} and Remark \ref{system_simplification}, 
the reader may safely return a value $v$ with tag $t$ when the number of tags strictly higher than $t$, denoted by $u$, satisfies $u \geq \nu$.
\end{remark}

In practice, if a read operation aborts, the reader in Algorithm \ref{sw_algo} can repeatedly invoke new read operations until it can decode a value that it can safely return. 
Due to asynchrony, a read may sample symbols from different writes at different times, and consequently, a read may not be able to see $k$ matching pieces of any single new value for indefinitely long, as long as new values continue to be written concurrently with the read. Therefore, we require reads to return in executions where a finite number of writes are invoked, thus only guaranteeing finite-write (FW) termination.

In Figure \ref{FW_Alg_1_variant}, we describe the read protocol for FW termination. We define a read \emph{iteration} to be an execution of Lines \ref{start_iteration} through \ref{end_iteration} in Figure \ref{FW_Alg_1_variant}. Different from the read protocol in Algorithm \ref{sw_algo}, a read invoking multiple iterations can combine responses from different iterations to form a quorum of responses that would allow it to reconstruct a value it can return. This is reflected in Line \ref{form_quorum}. Moreover, a read is allowed to return any value with tag that is at least as large as the highest tag it observed in the first iteration (Lines \ref{gamma_def} and \ref{algo1_line_13}). 
The \emph{write}$\_$\emph{back} procedure in Line \ref{same_write_back} of Figure \ref{FW_Alg_1_variant} is the same as in Algorithm \ref{sw_algo}.
\begin{figure}
 \noindent \rule{\textwidth}{0.2pt}
\begin{algorithmic}[1] 
\Statex \textbf{Read protocol}
\Statex  Let $L$ and $T$ be empty sets, $once = false, \Gamma = 0$.
\Repeat  
\label{start_iteration}
\State   Send query request $get$ to all servers, await pairs $(t,element)$ from a quorum.  
\State   Let $R$ be the set of response pairs. $L \leftarrow L \cup R$.
\If { 	  $once = false$ } \label{update_gamma_condition}
\State Let $\Gamma$ be the maximum tag in $R$;   $once \leftarrow true$.
\label{gamma_def}
\EndIf
\State Let $ \mathcal{R} =\{  R_1 ,\ldots,  R_{|\mathcal{R}|} \}$ such that, for each $i$, $ R_i  \subset L, | R_i | \geq N-f $ and $R_i $ does not contain  
\Statex \quad \quad more than a single response from the same server.
\label{form_quorum}
\State Let i=0.
  \While { ($T = \emptyset$  and  $i \neq | \mathcal{R} |$)} \label{while_loop_read_protocol}
  \State $i \leftarrow i+1$. Consider $R_i $.
  \State  Let $T$ be the set of decodable tags $t$ occurring in $R_i$ such that   
\State  $(i)$ $t$ appears in at least $f+1$ responses,  \label{algo1_line_11}
\State $(ii)$ Or, the number of tags strictly higher than $t$ is at most $\nu$,\label{algo1_line_12}
\State $(iii)$ Or, $t \geq \Gamma$.\label{algo1_line_13}
  \EndWhile 
 \Until {$T \neq \emptyset$}  \label{end_iteration}
  \State  Let $t = \max(T)$, and $v$ its value. 
   \State  \emph{write}$\_$\emph{back}$(t,v)$. \label{same_write_back}
\end{algorithmic}
 \noindent \rule{\textwidth}{0.2pt}
 \caption{\textbf{Modified read protocol for FW termination.} }
 \label{FW_Alg_1_variant}
\end{figure}
\begin{lemma}[Finite-write termination]
\label{FW_termination_SW_algo}
Algorithm $1$ with re-invoked reads as in Figure \ref{FW_Alg_1_variant} guarantees atomicity and FW termination. 
\end{lemma}
\begin{proof}
Let $t_1$ denote the tag of all preceding writes and reads and $Q_1$ denote the quorum of servers that replied to the \emph{put} message of the latest finished operation. Let $R$ denote the quorum of servers that replied to the read in its first iteration. As $R \cap Q_1 \neq \emptyset$, it follows that $\Gamma \geq t_1$. That is, the tag $\Gamma$, selected in Line \ref{gamma_def}, is higher than or equal to the tag of all preceding writes and reads. Thus, a value returned by a reader with tag satisfying Line \ref{algo1_line_13} in Figure \ref{FW_Alg_1_variant} satisfies the statement of Lemma \ref{lem:tagsordering_read}. Note that $\Gamma$ is updated only once (Line \ref{update_gamma_condition}). Otherwise, suppose that Line \ref{algo1_line_13} is not satisfied, which means that the tag of the returned value is smaller than $\Gamma$. By the read protocol, there exists a quorum $Q $ from which the read obtained its value $v$. The value $v$ satisfies either Line \ref{algo1_line_11} or Line \ref{algo1_line_12}. The read operation can be seen to be \textit{equivalent} to another read operation which terminated in a single iteration and received responses from $Q$. Thus, Lemma \ref{lem:tagsordering_read} is satisfied in this case. Moreover, Lemmas \ref{lem:tagsordering_write} and \ref{lem:writeTagDifferentAlg1} hold. Thus, Algorithm \ref{sw_algo} with re-invoked reads as in Figure \ref{FW_Alg_1_variant} satisfies atomicity, with a similar proof to Theorem \ref{thm:atomicity}. 

Consider a fair execution with finite number of writes. 
Since non-failed writes terminate, let $P$ be a point that all writes either completed or failed. There can be at most $\nu-1$ failed writes (in fact since there is a single writer, there is at most 1 failed write/concurrent write). 
Suppose there is a read that does not complete by point $P$.
Because the execution is fair and at most $f$ servers fail, the read can take steps and restart an iteration. Hence the read is re-invoked after point $P$. There is at most $\nu -1$ concurrent writes (corresponding to failed writes) with the read, hence the read terminates by Theorem \ref{thm:readcompletes} and the fact that $\nu \ge 2$.
\end{proof}
\subsection{Storage and Communication Costs of Algorithm \ref{sw_algo}}
\label{cost_SW_algo}
By the design of Algorithm \ref{sw_algo}, the following result follows.
\begin{theorem}
The storage cost of Algorithm \ref{sw_algo} is 
$
\frac{N}{k} = \frac{N}{\lceil\frac{N-2f}{\nu}\rceil}
$ units.
The write communication cost  is $\frac{N}{k}$. The worst-case read communication cost is $\frac{2N}{k}$, including the {\emph{write$\_$back}} phase. 
\end{theorem}

Comparing the storage cost of Algorithm \ref{sw_algo} and ABD algorithm (we assume that ABD uses only $2f+1$ servers and the remaining servers are unused, to obtain the minimal storage cost $2f+1$), we have 
\begin{equation} \label{eq:storage_comparison_with_ABD}
2f+1 - \frac{N}{k} = \frac{(N - 2f - 1) (2f+1 -\nu)}{N - 2f + \nu - 1} = \frac{ k-1}{ k}(2f+1 -\nu).
\end{equation} 
Therefore, when $\nu \le  2f+1$, the storage of Algorithm \ref{sw_algo} is at most that of ABD algorithm using $2 f+ 1$ servers. 
Compared to ABD, Algorithm \ref{sw_algo} makes use of more available servers to offer storage reduction, at the expense of not guaranteeing liveness if a read operation is concurrent with at least $\nu$ write operations.

\begin{remark}
A reader does not need to write back a value $v$ with tag $t$ if the reader has observed a tag higher than $t$ among the responses. 
By the well-formedness assumption, the write client does not start a new write until it completes the previous one.
\end{remark}
\section{Multi-Writer Multi-Reader Algorithm}
\label{multi_writer_algorithm_1}
\subsection{Algorithm Description}
Assume in this section that $\nu < N-2f, \nu \geq 1$ and thus $k>1$. We extend Algorithm \ref{sw_algo} to the multi-writer setting, in which we assume the presence of an arbitrary number of writers, as well as an arbitrary number of readers. The proposed algorithm, referred to as Algorithm \ref{Algorithm_mw_1}, combines replication and multi-version codes while achieving consistency and low storage costs.
Algorithm  \ref{Algorithm_mw_1} contains a description of the protocol. Here, we provide an overview of the algorithm.

Each server maintains a $(tag, element )$ tuple, where \emph{element} can be a full replica or a coded symbol. We assume that tags are tuples of the form $(z, ‘id’)$, where $z$ is an integer and $‘id’$ is an identifier of a write client. The ordering on the set of tags $\mathcal{T}$ is defined lexicographically, using the usual ordering on the integers and a predefined ordering on the client identifiers. 
In the write protocol, the \emph{query} phase first obtains the tags of the servers, in order to generate a higher tag. 
In the \emph{pre-write} phase, a writer propagates a full replica to at least $k + f$ servers, to ensure that the consistency of the data is not compromised in the presence of concurrent
writers.
In the \emph{finalize} phase,  coded symbols are sent to a quorum of servers. The servers only maintain the highest tag, and the corresponding element.
As $k=\lceil \frac{N-2f}{\nu} \rceil$, it follows that $N-f\geq f+k + (\nu -1) k \geq f+k$.
Hence, the \emph{pre-write} quorum is smaller than the \emph{finalize} quorum.
The read protocol of Algorithm \ref{Algorithm_mw_1} is essentially similar to Algorithm \ref{sw_algo}, except that a reader can receive coded symbols and/or replicas. In particular, a decodable tag may come from a replica and/or $k$ matching coded symbols. 
The \emph{write}$\_$\emph{back} procedure in the read contains two phases: \emph{pre-write} and \emph{finalize} that are the same as the write protocol.
\begin{algorithm}
\caption{\textbf{: multi-writer multi-reader setting} }
\begin{algorithmic}[1] 
\Statex  \textbf{Write protocol} \newline
 \emph{Input:} Value $v$, {$v \in \mathcal{V}$}.
\Statex \underline{\emph{query}  }
\State Send \emph{get$\_$tag} messages to all servers asking for their stored tags; await responses from a quorum.
\Statex \underline{\emph{pre-write}  }
\State Select the largest tag from the query phase; let its integer component be $z$. Form a new tag $t$ as $(z+1, ‘id’)$, where $‘id’$ is the identifier of the client performing the operation. Send $put(t,v)$ to server $s,$ for every $s \in \{1,2,\ldots,k+2f\}$.
\State Await acknowledgement from $k+f$ servers.

\Statex \underline{\emph{finalize} }
\State Use the $(N,k)$ code to get $N$ coded symbols. Denote $(y_1, y_2,\ldots, y_N) = \Phi(v)$.
 \State     Send $put(t,y_s)$ to server $s,$ for every $s \in \{1,2,\ldots,N\}$. Await acknowledgement from a quorum, and then terminate.
\newline  
\Statex \textbf{Read protocol}
\State	Send query request $get$ to all servers, await pairs $(t,element)$ from a quorum.  
\State	Let $R$ be the set of response pairs.
\State	Let $T$ be the set of decodable tags $t$ occurring in $R$ such that   
\State  $(i)$ $t$ appears in at least $f+1$ responses.
\label{algo2_line_1} 
\State $(ii)$ Or, the number of tags strictly higher than $t$ is at most $\nu$.
     \If{ $T \neq \emptyset$} 
     \State  Let $t = \max(T)$, and $v$ its value. 
     		 \State  \emph{write}$\_$\emph{back}$(t,v)$    		 
     \Else
     \State $\_$\emph{abort}$\_$
     \EndIf   
 \Procedure{ \emph{write}$\_$\emph{back} }{$ t,v $}
 \Statex \underline{\emph{pre-write}  }
\State Send $put(t,v)$ to server $s,$ for every $s \in \{1,2,\ldots,k+2f\}$. 
\State Await acknowledgement from $k+f$ servers.
\Statex \underline{\emph{finalize} }
\State Use the $(N,k)$ code to get $N$ coded symbols. Denote $(y_1, y_2,\ldots, y_N) = \Phi(v)$.
 \State     Send $put(t,y_s)$ to server $s,$ for every $s \in \{1,2,\ldots,N\}$. Await acknowledgement from a quorum, and then terminate.
\EndProcedure
\newline
\Statex \textbf{Server $s$ protocol}
   \Statex \emph{state variable:} A pair $(t,{x})$, where $t \in  \mathcal{T}$, $x \in \mathcal{V} \cup \mathcal{W}$. 
	  \Statex {\emph{initial state:} Store the default pair $(0, y_s)$, where $y_s$ is the $s$th component of $\Phi(v_0),$ where $v_0 \in \mathcal{V}$ is the default initial value.}  
  \State   On receipt of \emph{get$\_$tag}: respond with the stored tag.
    \State On receipt of \emph{get}: respond with $(t,*)$, where $*$ can be a coded symbol or a full value.   
    \State On receipt of {$put(t_{new},v_{new})$}, such that $v_{new} \in \mathcal{V}$: If $t_{new}>t$, then set {$t \gets t_{new}$ and ${x} \gets v_{new}.$} In any case respond with acknowledgement.    
   \State On receipt of {$put(t_{new},y_{new})$}, such that $y_{new} \in \mathcal{W}$:  If $t_{new}\geq  t$, then set {$t \gets t_{new}$ and ${x} \gets y_{new}.$} In any case respond with acknowledgement.
\end{algorithmic}
 \label{Algorithm_mw_1}
\end{algorithm}
\begin{remark}
\label{saving_communication_cost}
During the \emph{finalize} phase of the write protocol, the writer does not need to send coded symbols to 
all the servers. Indeed, the writer sends full replicas to the first $k+2f$ servers in the \emph{pre-write} phase. Then, during the \emph{finalize} phase, the writer sends coded symbols to the remaining $N-k-2f$ servers, and only a finalize message with the corresponding tag to the first $k+2f$ servers so as to minimize the communication cost.
\end{remark}
\begin{remark}
A reader's \emph{write$\_$back} proceeds in two phases: \emph{pre-write} and \emph{finalize}. If a reader has received at least one coded response with tag $t$, then, the \emph{pre-write} phase of tag $t$ must have already completed. The reader does not need to carry out the \emph{pre-write} step and it may only perform the \emph{finalize} phase. 
Moreover, the \emph{write$\_$back} procedure can be entirely skipped if the read observes $N-f$ coded symbols with tag $t$.
\end{remark}
\subsection{Safety Properties}
\begin{definition}[Tag of an operation $\pi$]
We define tags of operations in the same way as in Definition \ref{def:tag_definition}. The put message used to define the tag can be either from a pre-write or a finalize phase since they have the same tag.
\end{definition}
\begin{lemma}[persistence of data]
\label{persistence_data_mw_1}
The value with the highest tag can be fully recovered at any point of an execution, as long as the number of server failures is bounded by $f$.
\end{lemma}
\begin{proof}
We analyze the storage content of the system, at an arbitrary point $P$ of an execution.  Let $t$ be the maximum tag in the system, at point $P$. If $t$ is stored with a full replica at some server, then we can immediately recover the value with tag $t$. Otherwise, let $u$ be the number of the different coded symbols with tag $t$ that are stored in the system. Since there is no full replica with tag $t$, it follows from the write protocol that the writer of tag $t$ has finished its \emph{pre-write} phase and started its \emph{finalize} phase. Thus, the value with tag $t$ has been stored in at least $k+f$ servers. Moreover, these replicas must have been replaced by their corresponding coded symbols, or are stored in failed servers. Finally, noting that at most $f$ servers can fail, it follows that $u \geq (k+f)-f=k$. This means that there exists a sufficient number of coded symbols that allow recovery of the value with tag $t$.
\end{proof}
 The proof of atomicity follows along the lines of the single-writer algorithm, with appropriate modifications. In particular, the statements of Lemmas \ref{lem:a_safety_property}, 
\ref{lem:tagsordering_read} and \ref{lem:tagsordering_write} hold by considering the write quorum to be the \emph{finalize} quorum. The statement of Lemma \ref{lem:writeTagDifferentAlg1} follows from the \emph{query} phase. 
We start by stating the equivalent of Lemma \ref{lem:a_safety_property} in the context of Algorithm 2.
 
\begin{lemma}\label{lem1_mw_1}
Consider any execution $\alpha$ of Algorithm 2 and consider a write or read operation $\pi_1$ that completes in $\alpha$. Let $T(\pi_1)$ denote the tag of the operation $\pi_1$ and let $Q_1$ denote the quorum of servers from which responses are received by $\pi_1$ to its \emph{finalize} message.  Consider a read operation $\pi_r$ in $\alpha$ that is invoked after the termination of the write operation $\pi_1$. Suppose that the read $\pi_r$ receives responses to its \emph{get} message from a quorum $Q_r$. Then, \\
(1) Every server $s$ {in} $Q_1 \cap Q_r$ responds to the \emph{get} message from $\pi_r$ with a tag that is at least as large as $T(\pi_1)$.\\
(2) If, among the responses to the \emph{get} message of $\pi_r$ from the servers in $Q_{1} \cap Q_{r}$, the number of tags is at most {$\nu$}, then there is some tag $t$ such that  (i) $t \geq T(\pi_1)$, 
	  and  (ii) from the servers in $Q_1 \cap Q_r$, $\pi_r$ receives a full replica or $k$ coded symbols with tag $t$.
  \label{lem:a_safety_property_mw_1}
\end{lemma}
\begin{proof}
\noindent \emph{Proof of (1).} Similar to the proof of (1) in Lemma \ref{lem:a_safety_property}.\\
\noindent \emph{Proof of (2).} Among the responses from $Q_1 \cap Q_r$, the read $\pi_r$ receives at most $\nu$ different tags. If among these responses, the reader receives a full replica, then, the Lemma follows. Otherwise, all the responses from $Q_1 \cap Q_r$ are associated with coded symbols. The rest of the proof is similar to Lemma \ref{lem1}.
\end{proof}
 
 \begin{lemma}
{Consider an execution $\alpha$ of Algorithm 2. Let $\pi_1$ be a write or read operation that completes in $\alpha$, and let $\pi_2$ be a read operation that completes in $\alpha$. }Let $T(\pi_1)$ denote the tag of operation $\pi_1$ and $T(\pi_2)$ denote the tag of operation $\pi_2$. If $\pi_2$ begins after the termination of $\pi_1$, then $T(\pi_2) \geq T(\pi_1)$.
	\label{lem:tagsordering_read_mw_1}
\end{lemma}
\begin{proof}
Similar to Lemma \ref{lem:tagsordering_read}. 
\end{proof}
Remark \ref{multiple_return_values} holds in the context of Algorithm \ref{Algorithm_mw_1}. In particular, a reader can safely return a value for which it has received at $f+1$ responses (which can be coded or non-coded).
\begin{lemma}
{Consider an execution $\alpha$ of Algorithm 2. Let $\pi_1$ be a write or read operation that completes in $\alpha$, and $\pi_2$ be a write operation that completes in $\alpha$.}  Let $T(\pi_1)$ denote the tag of operation $\pi_1$ and $T(\pi_2)$ denote the tag of operation $\pi_2$.  If $\pi_2$ begins after the termination of $\pi_1$, then $T(\pi_2) > T(\pi_1)$. 
	\label{lem:tagsordering_write_mw_1}
\end{lemma}
\begin{proof}
The operation $\pi_1$ terminates after completing its \emph{finalize} phase, during which it receives responses from a quorum, $Q_f(\pi_1)$. From the server protocol, we can observe that every server $s$ in $Q_f(\pi_1)$ stores a tag that is at least as large as $T(\pi_1)$ at the point of responding to the second \emph{put} message of $\pi_1$.
 We denote the quorum of servers that respond to the \emph{query} phase of $\pi_2$ as $Q_q(\pi_2)$. It follows that $\pi_2$ receives tags that are no smaller than $T(\pi_1)$ from every server $s \in Q_f(\pi_1) \cap Q_q(\pi_2)$. Because the integer part is incremented, the largest integer part of the tag $z$ that the writer in $\pi_2$ observes is no less than $ T(\pi_1)$. It follows that $T(\pi_2) > T(\pi_1) $.
\end{proof}

\begin{lemma}
\label{lem:tagsdifferentwriters}
Let $\pi_1$, $\pi_2$ be write operations that terminate in an execution $\alpha$ of Algorithm 2. Then $T(\pi_1) \neq T(\pi_2)$.
\end{lemma}
\begin{proof}
Let $\pi_1$, $\pi_2$ be write operations that terminate in an execution $\alpha$. Let $C_1$, $C_2$ respectively indicate the identifiers of the client nodes at which operations $\pi_1$, $\pi_2$ are invoked. We consider two cases.

\noindent\emph{Case 1}: $C_1 \neq C_2$: From the write protocol, we note that $T(\pi_i)$ = $(z_i,C_i)$. Since $C_1 \neq C_2$ we have $T(\pi_1) \neq T(\pi_2)$.\\
\noindent\emph{Case 2}: $C_1 = C_2$: Recall that operations at the the same clients are \emph{well-formed}: where a new invocation awaits the response of a preceding invocation. It means that one of the operations should complete before the other starts. Suppose that, without loss of generality, the write operation $\pi_1$ completes before the write operation $\pi_2$ starts. Then, Lemma \ref{lem:tagsordering_write_mw_1} implies that $T(\pi_2) > T(\pi_1)$. This implies that $T(\pi_2) \neq T(\pi_1) $.
\end{proof}

\begin{theorem}
\label{thm:atomicity_mw_1} 
Algorithm \ref{Algorithm_mw_1} emulates an atomic read-write object.
\end{theorem}
\begin{proof}
Same steps as in Theorem \ref{thm:atomicity}, with Lemmas 
\ref{lem:tagsordering_read}, \ref{lem:tagsordering_write} and 
\ref{lem:writeTagDifferentAlg1}
replaced by Lemmas 
\ref{lem:tagsordering_read_mw_1}, \ref{lem:tagsordering_write_mw_1} and  \ref{lem:tagsdifferentwriters}, respectively.
\end{proof}

%
\begin{remark}
A write operation $\pi_1$ needs to be counted as ``concurrent'' with the read $\pi_r$ only if its tag $T(\pi_1)$ is larger than $t_w$ defined in Lemma \ref{lem:a_secondsafety_property}. In particular, suppose $\pi_1$ is a failed write operation, then it is not counted as concurrent with the read unless $T(\pi_1) > t_w$.
\end{remark}
The statement of the safety property in Lemma \ref{lem:a_secondsafety_property} holds also in the setting of Algorithm \ref{Algorithm_mw_1}, and can also be used to prove the liveness properties of Algorithm \ref{Algorithm_mw_1}. 

\begin{lemma} \label{lem2_mw_1}
Consider any execution $\alpha$ of Algorithm \ref{Algorithm_mw_1}. {Let $\pi_r$ denote a read operation in $\alpha$ that receives a quorum $Q_r$ of responses to its \emph{get} message.} Let $\mathcal{S}$ denote the set of all writes that terminate before the invocation of $\pi_r$ in $\alpha$. {If $\mathcal{S}$ is non-empty,} let $t_w$ denote the largest among the tags of the operations in $\mathcal{S}$. {If $\mathcal{S}$ is empty, let $t_w=0$}. 

If the number of writes concurrent with the read $\pi_r$ in $\alpha$ is smaller than $\nu$, then there is some tag ${t}$ such that  \\ (1) ${t} \geq t_w$,\\ (2) $\pi_r$ can recover the value with tag ${t}$, and
\\ (3) the number of tags that are higher than $t$ is smaller than $\nu$.
\label{lem:a_secondsafety_property_mw}
\end{lemma}
\begin{proof}
Similar to Lemma \ref{lem:a_secondsafety_property}.
\end{proof}
\subsection{Liveness Properties}
The proofs of the liveness properties of Algorithm \ref{Algorithm_mw_1} are similar to those of Algorithm 1, and are omitted.
\begin{theorem}[Termination of writes]
Consider any fair execution $\alpha$ of Algorithm \ref{Algorithm_mw_1} where the number of server failures is at most $f$, and the write client does not fail. Then, every write operation terminates in $\alpha$.
\label{thm:writecompletes_mw_1}
\end{theorem}

\begin{theorem}[Termination of reads]
Consider any fair execution $\alpha$ of Algorithm \ref{Algorithm_mw_1} where the number of server failures is at most $f$. Consider any read operation that is invoked at a non-failing client in $\alpha$. If the number of writes concurrent with the read is strictly smaller than $\nu$, and the read client does not fail, then the read operation completes in $\alpha$.
\label{thm:readcompletes_mw_1}
\end{theorem}

Similar to Algorithm \ref{sw_algo}, if a read operation aborts, the reader in Algorithm \ref{Algorithm_mw_1} can repeatedly invoke new read operations until it can decode a value that it can safely return. When the read operation takes many iterations, we use the same read protocol as in Figure \ref{FW_termination_SW_algo}. Next we consider FW termination of a read. The challenge lies in the fact that even if there are finitely many writes, there can be infinitely many read operations that write back, preventing the read from satisfying Lines \ref{algo1_line_11}, \ref{algo1_line_12} or \ref{algo1_line_13}  in Figure \ref{FW_termination_SW_algo} in any iteration.
In our proof, we make use of Lemma \ref{persistence_data_mw_1}.
\begin{lemma}[Finite-write termination]
\label{FW_termination_MW_algo_1}
Algorithm \ref{Algorithm_mw_1} with re-invoked reads as in Figure \ref{FW_Alg_1_variant} guarantees atomicity and FW termination. 
\end{lemma}
\begin{proof}
The proof of atomicity is similar to Lemma \ref{FW_termination_SW_algo}. We present the proof for the termination of a read operation.

Consider a fair execution with a finite number of writes. Consider a read operation $\pi_r$. For $i \geq 1$, let $s_i$ denote the point such that $\pi_r$ initialed iteration $i$ and $t_i$ denote the point that $\pi_r$ received responses from a quorum of servers, denoted $Q_i$.

Assume that $\pi_r$ did not terminate at the end of iteration $i$ for some $i \geq 1$.  
If $Q_{i}= Q_{i+1}$ and the responses from each server $s \in Q_{i}= Q_{i+1}$ are the same in both iterations, then, it follows that each server $s \in Q_{i}= Q_{i+1}$ has not changed its state, and hence its response, between points $t_{i}$ and $s_{i+1}$. Thus, $\pi_r$ has observed an instantaneous image of the system at point $t_i$. By Lemma \ref{persistence_data_mw_1}, using its responses from $Q_i$, $\pi_r$ could have recovered a value with tag satisfying Line \ref{algo1_line_12} in Figure \ref{FW_termination_SW_algo}. Henceforth, $\pi_r$ could have terminated, a contradiction. It follows that either $Q_{i} \neq Q_{i+1}$ or  $Q_{i}= Q_{i+1}$ and there exists a server $s \in Q_{i}= Q_{i+1}$ that has responded with different $(tag, element)$ pairs in $Q_{i}$ and $Q_{i+1}$.

Note that each server can only increase the stored tag, or change from a replica to a coded symbol for the same tag. Because we assume an execution with a finite number of writes, each server may change its stored $(tag, element)$ pair, and thus its responses, only a finite number of times. Moreover, as there are finitely many quorum sets, there must exist an iteration $j\geq 1$ such that if $\pi_r$ failed in iteration $j$, then $Q_{j}=Q_{j+1}$, and each server in $s \in Q_{j}= Q_{j+1}$ replies with the same $(tag, element)$ pair. By the previous argument, $\pi_r$ will terminate at the end of iteration $j+1$. Moreover, $j$ is finite and satisfies $j \le  2( \binom{N}{N-f}+ \ldots+\binom{N}{N } )  N_w =   (2^{N+1}- 2^{N-f }) N_w$, where $N_w$ is the number of writes during the execution.
\end{proof}

\subsection{Storage and Communication Costs of Algorithm \ref{Algorithm_mw_1}}
\begin{theorem}
The worst-case storage cost of Algorithm \ref{Algorithm_mw_1} among all points in all executions corresponds to $k+2f+\frac{N-k-2f}{k}$. The steady-state storage cost is given by $\frac{N}{k}$. The write communication cost is $k+2f+\frac{N-k-2f}{k}$.
The worst-case read communication cost is $ 2(k+2f + \frac{N-k-f}{k})$, including the \emph{write$\_$back} phases. Here $k=\lceil \frac{N-2f}{\nu} \rceil$.
\end{theorem}

\begin{proof}
The first $k+2f$ servers stores a full replica in the worst case and stores a coded symbol in the steady state; the remaining servers stores a coded symbol. Hence the storage cost results hold.

A write operation proceeds in three phases. As we do not account for the cost of tags, the cost of a write operation is dominated by the cost of its \emph{pre-write} and \emph{finalize} phases. 
A writer sends its uncoded value to the first $k+2f$ servers in the \emph{pre-write} phase. Based on Remark \ref{saving_communication_cost}, in the \emph{finalize} phase, the writer needs to send coded symbols to only $N-k-2f$ servers. In total, the write communication cost is at most
$
k+2f + \frac{N-k-2f}{k}
$. 
The worst-case read cost corresponds to twice the write communication cost and it corresponds to a situation in which a reader needs to write back its value.
\end{proof}

We note that the worst-case storage of Algorithm \ref{Algorithm_mw_1} is incurred during each write operation. Moreover, one can see that while Algorithm 2 has higher worst-case storage than ABD, but its steady-state storage outperforms ABD when $\nu \le 2f+1$ (cf. Equation \eqref{eq:storage_comparison_with_ABD}).

\subsection{Algorithm \ref{Algorithm_mw_1}-A: Algorithm with Asymptotically Optimal Storage Cost}
\label{finite_writers}
In this section, we assume that $\nu \geq 2$ and $k>1$. We present a multi-writer algorithm that is similar to Algorithm \ref{Algorithm_mw_1}, except that the write procedure is limited to two phases. We further assume throughout this section that during all executions of the system, the following condition is satisfied:

\noindent \textbf{Condition 1:} \emph{at any point during the execution, the number of concurrent writes is smaller than $\nu$}. 

\noindent\textbf{Algorithm \ref{Algorithm_mw_1}-A description: }
By virtue of Condition 1, the write protocol in Algorithm \ref{Algorithm_mw_1} can be simplified as the need for a \emph{pre-write} phase is obviated. After acquiring its tag, a writer propagates coded symbols to a quorum of servers before terminating. 
The read and server protocols are exactly the same as in Algorithm \ref{sw_algo}. 

We note that if the number of writers is less than $\nu$ (e.g., in applications with pre-determined write clients), then Condition 1 is ensured by default. In particular, the single-writer setting is a special case of Condition 1. That is the reason that Algorithm \ref{Algorithm_mw_1}-A is almost the same as Algorithm \ref{sw_algo}. 
%
%

The proof techniques used for Algorithm \ref{sw_algo} hold for Algorithm 2-A. For instance, the \textit{persistence of data} property follows from Lemma \ref{lem2}. The only difference lies in showing that any two write operations have distinct tags, which follows from the \emph{query} phase of 
Algorithm 2-A. The proofs for Algorithm \ref{sw_algo} do not depend on the identity of the writer and the same reasoning holds. Hence, atomicity and liveness, i.e., Theorems \ref{thm:atomicity_mw_1}, 
\ref{thm:writecompletes_mw_1} and \ref{thm:readcompletes_mw_1} hold for Algorithm 2-A.
\begin{theorem}
The storage cost of Algorithm \ref{Algorithm_mw_1}-A is $\frac{N}{\lceil \frac{N-2f}{\nu} \rceil}$ at any point in any execution, which is asymptotically optimal for $N\gg f, \nu \le f+1$.
\label{thm:storage_cost_2_A}
\end{theorem}
 \begin{proof}
The storage cost at any point is $\frac{N}{\lceil \frac{N-2f}{\nu} \rceil}= \frac{N}{1+\frac{N-2f-1}{\nu}}=\frac{\nu}{N-2f+\nu-1}$ by Equation (\ref{N_value}). Therefore, when $N \gg f$, the storage cost is $\frac{\nu}{N-2f+\nu-1} \approx \frac{\nu}{N- f+\nu-1}$, which is the lower bound in \cite[Theorem 6.5]{cadambe2016information} for $\nu \le f+1$.
 \end{proof}
For executions such that there are less than $\nu$ writes at any point, a worst-case storage lower bound of $\Omega(min(f,\nu))$ is given in \cite{spiegelman2016space}  under lock-freedom, which is a weaker liveness property than 
$\nu$-concurrency wait-freedom. Algorithm \ref{Algorithm_mw_1}-A also meets this bound when $N \gg f, \nu \le f+1$. Other algorithms have also been proposed that match this bound. For example, the server can store all the concurrent coded symbols with coding parameter $N-2f$, achieving the worst-case storage of $\frac{\nu N}{N-2f}$ (e.g. \cite{cadambe2017coded}, and \cite{spiegelman2016space} for $N \gg f, v \le f+1$). But Algorithm \ref{Algorithm_mw_1}-A has a smaller multiplicative constant by a factor of up to $2$, which can significantly reduce the cost for systems such as memory-based data stores.

\section{Discussion and Conclusion} \label{sec:discussion}
Different from previous erasure code-based algorithms, our algorithms use a coding parameter $k$ that is determined by the liveness parameter $\nu$, given by $k = \lceil \frac{N-2f}{\nu} \rceil$. A system designer can choose $\nu$, and our algorithms guarantee the desired $\nu$-concurrency wait-freedom.
Due to the choice of this coding parameter, our algorithms store only one version of the coded symbol of size $\frac{1}{k}$ at each server. On the other hand, most previous erasure code-based algorithms use a larger coding parameter (typically $N-2f$), regardless of $\nu$, resulting in a smaller coded symbol per version, but more versions at each server. We discuss now the merits and the disadvantages of our approach.

We compare the storage cost of Algorithm \ref{sw_algo} with the cost incurred by previous algorithms.
When $\nu \le 2f+1$, the storage of Algorithm \ref{sw_algo} is no more than that of ABD (cf. Equation \eqref{eq:storage_comparison_with_ABD}). We also compare with the storage cost of the single-writer version of CASGC \cite{cadambe2017coded}, parameterized by $(k_{CASGC},\delta)=(N-2f, \nu -1)$, such that it offers the same liveness guarantees as Algorithm \ref{sw_algo}. The steady-sate storage cost of CASGC is shown to be $\frac{\nu N}{N-2f}$ and its worst-case storage is  $\frac{(\nu +1) N}{N-2f}$.
Interestingly, Algorithm \ref{sw_algo} can be up to twice as storage-efficient as CASGC for all values of the parameters $N, f, \nu$. Moreover, Algorithm \ref{sw_algo} has a simpler protocol structure than CASGC.

\begin{example}
For any $N,f$, let $\nu = N-2f-1$. Then, Algorithm \ref{sw_algo} uses coding parameter $k= 1 + \frac{N-(2f+1)}{\nu}=2$. The overall storage of Algorithm \ref{sw_algo} is then $\frac{N}{2}$. The steady-state storage of CASCG with parameters $(k_{CASGC},\delta)=(N-2f, N-2f -2)$ is $\frac{\nu N}{N-2f} = N(1-\frac{1}{N-2f})$, which goes to $N$ as $N-2f$ increases. Therefore, Algorithm \ref{sw_algo} can be close to twice as efficient as CASGC. Moreover, for the same choice of $\nu=N-2f-1$, whenever $\nu < 2f+1 \iff N <  4f +2 $, Algorithm \ref{sw_algo} improves upon ABD in terms of storage.
\end{example}

In addition, we compare multi-writer algorithms and assume an arbitrary number of concurrent writes during an execution.
We first compare Algorithm \ref{Algorithm_mw_1} to CASGC \cite{cadambe2017coded} with parameters $(k_{CASGC},$ $\delta)= (N-2f,\nu-1)$. 
Under these parameters, CASGC guarantees liveness of a read operation if the read is concurrent with at most $\nu-1$ write operations \cite[Theorem 4]{cadambe2017coded}, which matches Algorithm 2.
CASCG has a steady-state storage given by $\frac{\nu N}{N-2f}$. However, the worst-case storage of CASCG can be unbounded.
Algorithm \ref{Algorithm_mw_1} is advantageous compared to CASGC. It offers:
(1) smaller steady-state storage size given by $N/\lceil \frac{N-2f}{\nu} \rceil$, which can be close to twice as efficient, and (2) bounded worst-case storage.

Next we compare Algorithm \ref{Algorithm_mw_1} with SCCK \cite{spiegelman2016space}, which is an adaptive multi-writer multi-reader algorithm that combines replication and erasure codes. Assume its coding parameter is $k_{SCCK}=N-2f$, such that the steady-state storage is minimized.
The worst-case storage of SCCK is $2N$, which is at least twice the storage of ABD and Algorithm \ref{Algorithm_mw_1}. Meanwhile, SCCK has a steady-state storage $\frac{N}{N-2f}$, which is lower than Algorithm \ref{Algorithm_mw_1}. Note that SCCK only provides finite-write termination guarantees.  

While the steady-state storage of the multi-writer algorithms presented in this paper is higher compared to some other coding-based schemes, 
the in-place update for each write operation and the simple structure of the algorithms make them appealing from practical perspective and easy to implement.

In conclusion, we proposed fault-tolerant algorithms for emulating a shared memory over an asynchronous, distributed message-passing network. We first presented a single-writer multi-reader atomic shared memory algorithm, which forms the basis of our multi-writer multi-reader algorithms. 
Our algorithms guarantee liveness of the read as long as the number of writes concurrent with the read is smaller than a design parameter $\nu$. The parameter $\nu$ illustrates the tradeoff between liveness of read operations and the storage size per node. The overall steady-state storage and communication costs of our algorithms outperform ABD when $\nu < 2f+1$. 
An open problem is how to dynamically adapt to the concurrency level when it changes over time, so that the liveness condition is strengthened and the steady-state storage cost is lowered. Another interesting direction is to study the possibility of coding across several data objects in order to reduce the overall storage.
%
\bibliographystyle{abbrv}
\bibliography{biblio}

\begin{thebibliography}{10}

\bibitem{memcached}
Memcached: A distributed memory object caching system.
\newblock https://memcached.org/.

\bibitem{abraham2006byzantine}
I.~Abraham, G.~Chockler, I.~Keidar, and D.~Malkhi.
\newblock Byzantine disk paxos: optimal resilience with byzantine shared
  memory.
\newblock {\em Distributed Computing}, 18(5):387--408, 2006.

\bibitem{ABD}
H.~Attiya, A.~Bar-Noy, and D.~Dolev.
\newblock Sharing memory robustly in message-passing systems.
\newblock In {\em Proceedings of the ninth annual ACM symposium on Principles
  of distributed computing}, PODC '90, pages 363--375, New York, NY, USA, 1990.
  ACM.

\bibitem{CT}
C.~Cachin and S.~Tessaro.
\newblock Optimal resilience for erasure-coded byzantine distributed storage.
\newblock In {\em Dependable Systems and Networks, 2006. DSN 2006.
  International Conference on}, pages 115--124. IEEE, 2006.

\bibitem{cadambe2017coded}
V.~R. Cadambe, N.~Lynch, M.~Médard, and P.~Musial.
\newblock A coded shared atomic memory algorithm for message passing
  architectures.
\newblock In {\em IEEE 13th International Symposium on Network Computing and
  Applications}, pages 253--260, Aug 2014.

\bibitem{cadambe2016information}
V.~R. Cadambe, Z.~Wang, and N.~Lynch.
\newblock Information-theoretic lower bounds on the storage cost of shared
  memory emulation.
\newblock In {\em Proceedings of the 2016 ACM Symposium on Principles of
  Distributed Computing}, pages 305--313. ACM, 2016.

\bibitem{chockler2017space}
G.~Chockler and A.~Spiegelman.
\newblock Space complexity of fault tolerant register emulations.
\newblock In {\em Proceedings of the 2017 ACM Symposium on Principles of
  Distributed Computing}, pages 83--92. ACM, 2017.

\bibitem{dobre_powerstore}
D.~Dobre, G.~Karame, W.~Li, M.~Majuntke, N.~Suri, and M.~Vukoli{\'c}.
\newblock Po{W}er{S}tore: proofs of writing for efficient and robust storage.
\newblock In {\em Proceedings of the 2013 ACM SIGSAC conference on Computer \&
  communications security}, pages 285--298. ACM, 2013.

\bibitem{dutta2008optimistic}
P.~Dutta, R.~Guerraoui, and R.~R. Levy.
\newblock Optimistic erasure-coded distributed storage.
\newblock In {\em International Symposium on Distributed Computing}, pages
  182--196. Springer, 2008.

\bibitem{fan2003efficient}
R.~Fan and N.~Lynch.
\newblock Efficient replication of large data objects.
\newblock In {\em International Symposium on Distributed Computing}, pages
  75--91. Springer, 2003.

\bibitem{guerraoui2008collective}
R.~Guerraoui, R.~R. Levy, B.~Pochon, and J.~Pugh.
\newblock The collective memory of amnesic processes.
\newblock {\em ACM Transactions on Algorithms (TALG)}, 4(1):12, 2008.

\bibitem{HGR}
J.~Hendricks, G.~R. Ganger, and M.~K. Reiter.
\newblock Low-overhead byzantine fault-tolerant storage.
\newblock {\em ACM SIGOPS Operating Systems Review}, 41(6):73--86, 2007.

\bibitem{herlihy1990linearizability}
M.~P. Herlihy and J.~M. Wing.
\newblock Linearizability: A correctness condition for concurrent objects.
\newblock {\em ACM Transactions on Programming Languages and Systems (TOPLAS)},
  12(3):463--492, 1990.

\bibitem{konwar2016storage}
K.~M. Konwar, N.~Prakash, E.~Kantor, N.~Lynch, M.~M{\'e}dard, and A.~A.
  Schwarzmann.
\newblock Storage-optimized data-atomic algorithms for handling erasures and
  errors in distributed storage systems.
\newblock In {\em Parallel and Distributed Processing Symposium, 2016 IEEE
  International}, pages 720--729. IEEE, 2016.

\bibitem{konwar2016radon}
K.~M. Konwar, N.~Prakash, N.~Lynch, and M.~M{\'e}dard.
\newblock Radon: Repairable atomic data object in networks.
\newblock {\em arXiv preprint arXiv:1605.05717}, 2016.

\bibitem{konwar2017layered}
K.~M. Konwar, N.~Prakash, N.~Lynch, and M.~M{\'e}dard.
\newblock A layered architecture for erasure-coded consistent distributed
  storage.
\newblock {\em arXiv preprint arXiv:1703.01286}, 2017.

\bibitem{lakshman2010cassandra}
A.~Lakshman and P.~Malik.
\newblock Cassandra: a decentralized structured storage system.
\newblock {\em ACM SIGOPS Operating Systems Review}, 44(2):35--40, 2010.

\bibitem{lamport1986interprocess}
L.~Lamport.
\newblock On interprocess communication.
\newblock {\em Distributed computing}, 1(2):86--101, 1986.

\bibitem{lynch2002rambo}
N.~Lynch and A.~A. Shvartsman.
\newblock Rambo: A reconfigurable atomic memory service for dynamic networks.
\newblock In {\em International Symposium on Distributed Computing}, pages
  173--190. Springer, 2002.

\bibitem{Lynch1996}
N.~A. Lynch.
\newblock {\em Distributed Algorithms}.
\newblock Morgan Kaufmann Publishers Inc., San Francisco, CA, USA, 1996.

\bibitem{ousterhout2010case}
J.~Ousterhout, P.~Agrawal, D.~Erickson, C.~Kozyrakis, J.~Leverich,
  D.~Mazi\`{e}res, S.~Mitra, A.~Narayanan, G.~Parulkar, M.~Rosenblum, S.~M.
  Rumble, E.~Stratmann, and R.~Stutsman.
\newblock The case for {RAMClouds}: scalable high-performance storage entirely
  in {DRAM}.
\newblock {\em ACM SIGOPS Operating Systems Review}, 43(4):92--105, 2010.

\bibitem{patterson1988case}
D.~A. Patterson, G.~Gibson, and R.~H. Katz.
\newblock {\em A case for redundant arrays of inexpensive disks (RAID)},
  volume~17.
\newblock ACM, 1988.

\bibitem{spiegelman2016space}
A.~Spiegelman, Y.~Cassuto, G.~Chockler, and I.~Keidar.
\newblock Space bounds for reliable storage: Fundamental limits of coding.
\newblock In {\em Proceedings of the 2016 ACM Symposium on Principles of
  Distributed Computing}, pages 249--258. ACM, 2016.

\bibitem{wang2017multi}
Z.~Wang and V.~R. Cadambe.
\newblock Multi-version coding - an information-theoretic perspective of
  consistent distributed storage.
\newblock {\em IEEE Transactions on Information Theory}, PP(99):1--1, 2017.

\end{thebibliography}

\end{document}